\documentclass[10pt]{iopart}

\newcommand{\Rm}{\mathbb{R}}

\newcommand{\Nm}{\mathbb{N}}

\newcommand{\be}{\[}
\newcommand{\ee}{\]}
\newcommand{\ba}{\begin{eqnarray*}}
\newcommand{\ea}{\end{eqnarray*}}

\newcommand{\pp}{\partial}
\newcommand{\vv}[1]{\boldsymbol{\mathrm{#1}}}

\newcommand{\argmin}{\mathop{\mathrm{arg\,min}}}

\usepackage{iopams,setstack}
\usepackage{amsthm,amssymb,mathrsfs}
\usepackage{graphicx}

\newtheorem{thm}{Theorem}[section]
\newtheorem{lem}[thm]{Lemma}

\theoremstyle{remark}

\begin{document}

\title[]{The inverse Rytov series for diffuse optical tomography}

\author{Manabu Machida}
\address{Department of Informatics, Faculty of Engineering, 
Kindai University, Higashi-Hiroshima 739-2116, Japan}
\ead{\mailto{machida@hiro.kindai.ac.jp}}


\begin{abstract}
The Rytov approximation is known in near-infrared spectroscopy including diffuse optical tomography. In diffuse optical tomography, the Rytov approximation often gives better reconstructed images than the Born approximation. Although related inverse problems are nonlinear, the Rytov approximation is almost always accompanied by the linearization of nonlinear inverse problems. In this paper, we will develop nonlinear reconstruction with the inverse Rytov series. By this, linearization is not necessary and higher order terms in the Rytov series can be used for reconstruction. The convergence and stability are discussed. We find that the inverse Rytov series has a recursive structure similar to the inverse Born series.
\end{abstract}

\maketitle

\section{Introduction}

We consider diffuse light propagation in a bounded domain $\Omega\subset\Rm^n$ ($n\ge2$) with a smooth boundary $\pp\Omega$. In diffuse optical tomography, coefficients of the diffusion equation are determined from boundary measurements. In this paper, we consider the reconstruction of the absorption coefficient.

The time-independent diffusion equation is given by
\begin{equation}
\cases{
-D_0\Delta u+\mu_au=f,
&\quad $x\in\Omega$,
\\
D_0\pp_{\nu}u+\frac{1}{\zeta}u=0,
&\quad $x\in\pp\Omega$.
}
\label{de1}
\end{equation}
Here, $D_0,\zeta$ are positive constants, and $\pp_{\nu}$ denotes the directional derivative with the outward unit vector $\nu$ normal to $\pp\Omega$. Furthermore, $\mu_a$ is the absorption coefficient and $f$ is the source term.
The outgoing light $u(x)$ is detected on a subboundary $\Gamma$ of the boundary ($x\in\Gamma\subset\pp\Omega$). On the boundary, we suppose $u\in L^p(\Gamma)$ with some $p\ge1$.

Since the cost function for the inverse problem of determining coefficients of the diffusion equation in (\ref{de1}) has a complicated landscape, the reconstructed value is trapped in a local minimum if iterative schemes such as the Levenberg-Marqusrdt, Gauss-Newton, and conjugate gradient methods are used. An alternative approach is the use of direct methods in which perturbations of coefficients are reconstructed. The Born and Rytov approximations are frequently used in cooperation with linearization of the nonlinear inverse problem. When the (first) Born approximation is compared with the (first) Rytov approximation, the superiority of the latter has been discussed \cite{Arridge99,Keller69,Kirkinis08}.

A systematic way of inverting the Born series has been studied \cite{Markel-OSullivan-Schotland03,Markel-Schotland07,Moskow-Schotland08,Moskow-Schotland09,Panasyuk-Markel-Carney-Schotland06}. That is, higher-order Born approximations can be implemented with the inverse Born series. In this way, the direct methods can be applied to nonlinear inverse problems without linearization. In Ref.~\cite{Machida-Schotland15}, the inverse Born series was implemented for the transport-based optical tomography. In addition to optical tomography, the Calder\'{o}n problem was considered with the inverse Born series \cite{Arridge-Moskow-Schotland12}. The inverse Born series was applied to inverse problems for scalar waves \cite{Kilgore-Moskow-Schotland12} and for electromagnetic scattering \cite{Kilgore-Moskow-Schotland17}. The series was developed for discrete inverse problems \cite{Chung-Gilbert-Hoskins-Schotland17}. The technique of the inverse Born series was used to investigate the inversion of the Bremmer series \cite{Shehadeh-Malcolm-Schotland17}. The inverse Born series was extended to Banach spaces \cite{Bardsley-Vasquez14,Lakhal18}. Recently, a modified Born series with unconditional convergence was proposed and its inverse series was studied \cite{Abhishek-Bonnet-Moskow20}. The convergence theorem for the inverse Born series has recently been improved \cite{Hoskins-Schotland22}. See Ref.~\cite{Moskow-Schotland19} for recent advances. Moreover a reduced inverse Born series was proposed \cite{Markel-Schotland22}.

Based on the success of past studies on the inverse Born series, in this paper we consider the inversion of the Rytov series. In experimental and clinical researches on optical tomography, quite often the Born approximation is impractical and tomographic images are obtained with the Rytov approximation. After linearization, the Rytov approximation was used for detecting breast cancer \cite{Choe-etal05,Choe-etal09} and used when the brain function was studied through the neurovascular coupling \cite{Eggebrecht-etal14}. The limitation of the linear approximation has been pointed out \cite{Boas97}.

Indeed, the inverse Rytov series was considered for the Helmholtz equation and it was numerically observed that the inverse Rytov series with the first through third approximations give better reconstructed images than the inverse Born series \cite{Tsihrintzis-Devaney00}. In \cite{Marks06}, intermediate approximations between the Born and Rytov approximations were explored. The relation between the inverse Rytov series and Newton's method was investigated \cite{Park-etal11}. In these papers, however, no systematic way of computing higher-order terms was presented.

The remainder of the paper is organized as follows. The Born series is introduced in Sec.~\ref{born} and the Rytov series is introduced in Sec.~\ref{rytov}. Then the inverse Rytov series is discussed in Sec.~\ref{invrytov}. Section \ref{toymodel} is devoted to the implementation of the inverse Rytov series and numerical examples. Concluding remarks are given in Sec.~\ref{concl}.

\section{The Born series}
\label{born}

Let $g$ be a positive constant. We write
\be
\mu_a(x)=g\left(1+\eta(x)\right),\quad\eta\ge-1.
\ee
We suppose that $\eta$ is supported in a closed ball $B_a$ of radius $a$:
\be
\mathop{\mathrm{supp}}\eta\subset B_a\subset\Omega.
\ee
It will be seen below that the Born series converges for sufficiently small $a>0$. We suppose that $\eta\in L^q(B_a)$ for some $q\ge2$.

Let $u_0(x)$ be the solution to the equation (\ref{de1}) in which $\mu_a(x)$ is replaced by $g$. We assume that there exists a constant $\xi>0$ such that $\xi\le u_0$ on $\Gamma$. Let $G(x,y)$ be the Green's function which corresponds to $u_0$. Then the following identity holds.
\be
u(x)=u_0(x)-g\int_{\Omega}G(x,y)\eta(y)u(y)\,dy.
\ee
From the above identity, the Born series can be constructed as
\be
u=u_0+u_1+\cdots,
\ee
where
\be
u_j(x)=-g\int_{\Omega}G(x,y)\eta(y)u_{j-1}(y)\,dy\quad(j=1,2,\dots).
\ee
The first two terms of the Born series are obtained as
\ba
u_1(x)
&=
-g\int_{\Omega}G(x,y)\eta(y)u_0(y)\,dy,
\\
u_2(x)
&=
g^2\int_{\Omega}\int_{\Omega}G(x,y)\eta(y)G(y,z)\eta(z)u_0(z)\,dydz.
\ea

Let us introduce the multilinear operators $K_j:L^q(B_a)\times\cdots\times L^q(B_a)\to L^p(\Gamma)$ such that
\be
u_j=-K_j\eta^{\otimes j},
\ee
where $\eta^{\otimes j}=\eta\otimes\cdots\otimes\eta$ is the $j$-fold tensor product. Here we have
\ba
&\fl
K_1\eta=
g\int_{B_a}G(x,y)u_0(y)\eta(y)\,dy,
\\
&\fl
K_2\eta\otimes\eta=
-g^2\int_{B_a}\int_{B_a}G(x,y)G(y,z)u_0(z)\eta(y)\eta(z)\,dydz.
\ea
In general, the $j$th term is given by
\ba
K_j\eta^{\otimes j}
&=
(-1)^{j+1}g^j\int_{B_a\times\cdots\times B_a}G(x,y_1)G(y_1,y_2)\cdots
G(y_{j-1},y_j)
\\
&\times
u_0(y_j)\eta(y_1)\cdots\eta(y_j)\,dy_1\cdots dy_j.
\ea

Let us define the operators $\check{K}_j:L^q(B_a)\times\cdots\times L^q(B_a)\to L^{p}(\Gamma)$ such that
\be
\frac{1}{u_0}K_j\eta^{\otimes j}=\check{K}_j\eta^{\otimes j}.
\ee
We introduce
\be\fl
\mu=g\sup_{x\in B_a}\left\|G(x,\cdot)\right\|_{L^r(B_a)},\quad
\nu=g|B_a|^{1/r}\sup_{y_1,y_2\in B_a}\left\|G(\cdot,y_1)
\frac{u_0(y_2)}{u_0(\cdot)}\right\|_{L^p(\Gamma)},
\ee
where $r=q/(q-1)$.

\begin{lem}
For $j=1,2,\dots$, $\|\check{K}_j\|\le\nu\mu^{j-1}$.
\end{lem}

\begin{proof}
For any $f_i\in L^q(B_a)$ ($i=1,\dots,j$), the multilinear operators $\check{K}_j$ are written as
\ba
(\check{K}_jf_1\otimes\cdots\otimes f_j)(x)
&=
\frac{(-1)^{j+1}g^j}{u_0(x)}
\int_{B_a\times\cdots\times B_a}G(x,y_1)G(y_1,y_2)\cdots
G(y_{j-1},y_j)
\\
&\times
u_0(y_j)f_1(y_1)\cdots f_j(y_j)\,dy_1\cdots dy_j,
\quad x\in\Gamma.
\ea
Using H\"{o}lder's inequality, we have
\ba
&
\|\check{K}_jf_1\otimes\cdots\otimes f_j\|_{L^p(\Gamma)}^p
\\
&=
\left(g^j\right)^p\int_{\Gamma}\Biggl|\int_{B_a\times\cdots\times B_a}
G(x,y_1)G(y_1,y_2)\cdots G(y_{j-1},y_j)
\\
&\times
\frac{u_0(y_j)}{u_0(x)}f_1(y_1)\cdots f_j(y_j)\,dy_1\cdots dy_j\Biggr|^p\,dx
\\
&\le
g^{jp}\int_{\Gamma}\Biggl|
\left(\int_{B_a\times\cdots\times B_a}|f_1(y_1)\cdots f_j(y_j)|^q\,dy_1\cdots dy_j
\right)^{1/q}
\\
&\times
\left(\int_{B_a\times\cdots\times B_a}
\left|G(x,y_1)G(y_1,y_2)\cdots G(y_{j-1},y_j)\frac{u_0(y_j)}{u_0(x)}\right|^r
\,dy_1\cdots dy_j\right)^{1/r}
\Biggr|^p\,dx
\\
&\le
g^{jp}\|f_1\|_{L^q(B_a)}^p\cdots\|f_j\|_{L^q(B_a)}^p
\int_{\Gamma}\left|\sup_{y_1,y_j\in B_a}G(x,y_1)\frac{u_0(y_j)}{u_0(x)}\right|^p\,dx
\\
&\times
\left(\int_{B_a\times\cdots\times B_a}
\left|G(y_1,y_2)\cdots G(y_{j-1},y_j)\right|^r
\,dy_1\cdots dy_j\right)^{p/r}.
\ea
We define
\be
I_{j-1}=g^{j-1}\left(\int_{B_a\times\cdots\times B_a}
\left|G(y_1,y_2)\cdots G(y_{j-1},y_j)\right|^r\,dy_1\cdots dy_j\right)^{1/r}.
\ee
Similar to the calculation in \cite{Moskow-Schotland08}, we have
\be
I_{j-1}\le\mu I_{j-2},\quad I_1\le|B_a|^{1/r}\mu.
\ee
Hence,
\be
I_{j-1}\le\mu^{j-1}|B_a|^{1/r}\quad(j=2,3,\dots).
\ee
We obtain
\be
\|\check{K}_jf_1\otimes\cdots\otimes f_j\|_{L^p(\Gamma)}^p\le
\|f_1\|_{L^q(B_a)}^p\cdots\|f_j\|_{L^q(B_a)}^p\nu^p\mu^{p(j-1)}.
\ee
Therefore,
\be
\|\check{K}_j\|
=
\sup_{f_1,\dots,f_j\in L^q(B_a)\atop f_i\neq0\;(i=1,\dots,j)}
\frac{\left\|\check{K}_jf_1\otimes\cdots\otimes f_j\right\|_{L^p(\Gamma)}}
{\|f_1\|_{L^q(B_a)}\cdots\|f_j\|_{L^q(B_a)}}
\le\nu\mu^{j-1}.
\ee
\end{proof}

\section{The Rytov series}
\label{rytov}

Let us consider the Rytov series: $u=u_0e^{-\psi_1-\psi_2-\cdots}$. The function $\psi_j$ ($j=1,2,\dots$) is proportional to $g^j$. In particular, we consider boundary values of $u,u_0$ at $x\in\Gamma$. We introduce
\be
\psi=\psi(x)=\ln\frac{u_0(x)}{u(x)},\quad x\in\Gamma.
\ee
We assume $\psi\in L^p(\Gamma)$. We have
\ba
-\psi
&=
\ln\frac{u_0+u_1+\cdots}{u_0}=
\ln\left(1+\sum_{j=1}^{\infty}\frac{u_j}{u_0}\right)
\\
&=
\sum_{k=1}^{\infty}\frac{(-1)^{k+1}}{k}\left(\sum_{j=1}^{\infty}\frac{u_j}{u_0}\right)^k
\\
&=
\frac{u_1+u_2+\cdots}{u_0}-
\frac{(u_1+u_2+\cdots)^2}{2u_0^2}+
\frac{(u_1+u_2+\cdots)^3}{3u_0^3}-\cdots
\\
&=
-\psi_1-\psi_2-\cdots.
\ea
By collecting the first- and second-order terms, the first two terms of the Rytov series are explicitly written as
\be
\psi_1=-\frac{u_1}{u_0},\quad
\psi_2=-\frac{u_2}{u_0}+\frac{1}{2}\left(\frac{u_1}{u_0}\right)^2.
\ee
In general, we have
\be
\psi_j=\sum_{m=1}^j\frac{(-1)^m}{mu_0^m}\sum_{i_1+\cdots+i_m=j}u_{i_1}\cdots u_{i_m},\quad j=1,2,\dots.
\ee
We note that the number of $j$th order terms in $(u_1+\cdots)^m$ is
\be
\left(\begin{array}{c}j-1\\m-1\end{array}\right).
\ee
In total, the number of terms in $\psi_j$ is
\be
\sum_{m=1}^{j-1}\left(\begin{array}{c}j-1\\m-1\end{array}\right)=2^{j-1}.
\ee

We introduce the forward operators $J_j:L^q(B_a)\times\cdots\times L^q(B_a)\to L^p(\Gamma)$ such that
\be
\psi_j=J_j\eta^{\otimes j}\quad(j=1,2,\dots).
\ee
Note that $J_j$ are multilinear. We have
\ba
&\fl
J_1\eta=\check{K}_1\eta=
\frac{g}{u_0(x)}\int_{\Omega}G(x,y)u_0(y)\eta(y)\,dy,
\\
&\fl
J_2\eta\otimes\eta=
\check{K}_2\eta\otimes\eta+\frac{1}{2}\left(\check{K}_1\eta\right)^2
\\
&\fl=
\frac{g^2}{u_0(x)}\int_{\Omega}\int_{\Omega}G(x,y)G(y,z)u_0(z)\eta(y)\eta(z)\,dydz
+\frac{g^2}{2u_0(x)^2}\left(\int_{\Omega}G(x,y)u_0(y)\eta(y)\,dy\right)^2.
\ea
In general, the $j$th term is given by
\be
J_j\eta^{\otimes j}=\sum_{m=1}^j\frac{1}{m}\sum_{i_1+\cdots+i_m=j}
\left(\check{K}_{i_1}\eta^{\otimes i_1}\right)\cdots\left(\check{K}_{i_m}\eta^{\otimes i_m}\right).
\ee

\begin{lem}
We have $\|J_j\|\le\nu\left(\mu+\nu\right)^{j-1}$ for $j=1,2,\dots$. Moreover the Rytov series converges if $\|\eta\|_{L^q(B_a)}<(\mu+\nu)^{-1}$.
\end{lem}

\begin{proof}
We note the binomial formula:
\begin{equation}
x(x+y)^{j-1}=\sum_{m=1}^j\left(\begin{array}{c}j-1\\m-1\end{array}\right)x^my^{j-m}.
\label{binomial}
\end{equation}
We have
\ba
\|J_j\|
&\le
\sum_{m=1}^j\frac{1}{m}\sum_{i_1+\cdots+i_m=j}
\|\check{K}_{i_1}\|\cdots\|\check{K}_{i_m}\|
\\
&\le
\sum_{m=1}^j\left(\begin{array}{c}j-1\\ m-1\end{array}\right)\nu^m\mu^{j-m}
\\
&=
\nu\left(\mu+\nu\right)^{j-1}.
\ea

Since we have
\ba
\sum_{j=1}^{\infty}\|\psi_j\|_{L^p(\Gamma)}
&=
\sum_{j=1}^{\infty}\|J_j\eta\otimes\cdots\otimes\eta\|_{L^p(\Gamma)}
\le\sum_{j=1}^{\infty}\|J_j\|\|\eta\|_{L^q(B_a)}^j
\\
&\le
\nu\left(\mu+\nu\right)^{-1}
\sum_{j=1}^{\infty}\left(\mu+\nu\right)^j\|\eta\|_{L^q(B_a)}^j,
\ea
the series converges if $\|\eta\|_{L^q(B_a)}<(\mu+\nu)^{-1}$.
\end{proof}

\section{Inverse Rytov series}
\label{invrytov}

We begin by formally expanding the perturbation $\eta$ as
\ba
\eta
&=
\eta_1+\eta_2+\cdots
\\
&=
\mathcal{J}_1\psi+\mathcal{J}_2\psi\otimes\psi+\cdots.
\ea
We refer to the above series as the inverse Rytov series. If we substitute the series $\psi=J_1\eta+J_2\eta\otimes\eta+\cdots$, we have
\ba
&\fl
\eta=
\mathcal{J}_1\left(J_1\eta+J_2\eta\otimes\eta+\cdots\right)+
\mathcal{J}_2\left(J_1\eta+J_2\eta\otimes\eta+\cdots\right)\otimes\left(J_1\eta+J_2\eta\otimes\eta+\cdots\right)+\cdots
\\
&\fl=
\mathcal{J}_1J_1\eta+\left(\mathcal{J}_1J_2+\mathcal{J}_2J_1\otimes J_1\right)\eta\otimes\eta+\cdots.
\ea
Thus we obtain
\ba
&&
\mathcal{J}_1J_2+\mathcal{J}_2J_1\otimes J_1=0,
\\
&&
\mathcal{J}_3J_1\otimes J_1\otimes J_1+\mathcal{J}_2J_1\otimes J_2+
\mathcal{J}_2J_2\otimes J_1+\mathcal{J}_1J_3=0,\dots.
\ea

Indeed, the equality $\eta=\mathcal{J}_1J_1\eta$ does not hold due to the ill-posedness of this inverse problem. To consider $\mathcal{J}_1$, let us introduce $\eta^*$ as \cite{Machida-Schotland15}
\be
\eta^*=\argmin_{\eta\in B_a}\left(\frac{1}{2}\|J_1\eta-\psi\|_{L^p(\Gamma)}^2+\alpha R(\eta)\right),
\ee
where $R(\eta)$ is a penalty function with a regularization parameter $\alpha>0$ \cite{Engl-Hanke-Neubauer96,Morozov93,Schuster-Kaltenbacher-Hofmann-Kazimierski12}. The regularized pseudoinverse of $J_1$ is defined as $\mathcal{J}_1:\psi\mapsto\eta^*$. With this operator $\mathcal{J}_1$, we have
\be\fl
\mathcal{J}_2\psi\otimes\psi=
-\mathcal{J}_1J_2(\mathcal{J}_1\otimes\mathcal{J}_1)(\psi\otimes\psi)=
-\mathcal{J}_1\left[\check{K}_2(\mathcal{J}_1\otimes\mathcal{J}_1)(\psi\otimes\psi)+
\frac{1}{2}\left(\check{K}_1\mathcal{J}_1\psi\right)^2\right],
\ee
and
\ba\fl
\mathcal{J}_3\psi\otimes\psi\otimes\psi
=
-\left(\mathcal{J}_2J_1\otimes J_2+\mathcal{J}_2J_2\otimes J_1+\mathcal{J}_1 J_3\right)(\mathcal{J}_1\otimes\mathcal{J}_1\otimes\mathcal{J}_1)(\psi\otimes\psi\otimes\psi)
\\
\fl=
-\mathcal{J}_2(J_1\mathcal{J}_1\psi)\otimes\left[
\check{K}_2\mathcal{J}_1\psi\otimes\mathcal{J}_1\psi+\frac{1}{2}(\check{K}_1\mathcal{J}_1\psi)^2\right]
\\
\fl-
\mathcal{J}_2\left[\check{K}_2\mathcal{J}_1\psi\otimes\mathcal{J}_1\psi+
\frac{1}{2}(\check{K}_1\mathcal{J}_1\psi)^2\right]\otimes J_1\mathcal{J}_1\psi
\\
\fl-
\mathcal{J}_1\left[\check{K}_3(\mathcal{J}_1\psi)^{\otimes 3}+
(\check{K}_1\mathcal{J}_1\psi)(\check{K}_2\mathcal{J}_1\psi\otimes\mathcal{J}_1\psi)
+\frac{1}{3}(\check{K}_1\mathcal{J}_1\psi)^3\right].
\ea
For $j\ge2$, we have
\be
\mathcal{J}_j=-\left(\sum_{m=1}^{j-1}\mathcal{J}_m\sum_{i_1+\cdots+i_m=j}
J_{i_1}\otimes\cdots\otimes J_{i_m}\right)
\mathcal{J}_1\otimes\cdots\otimes\mathcal{J}_1.
\ee

\begin{thm}
\label{thm4_1}
Assume that there exists a constant $M_1<1$ such that $(\mu+2\nu)\|\mathcal{J}_1\|\le M_1$. Then the operator $\mathcal{J}_j:L^p(\Gamma)\times\cdots\times L^p(\Gamma)\to L^q(B_a)$ is bounded and
\be
\|\mathcal{J}_j\|\le C_1\left(\mu+2\nu\right)^j\|\mathcal{J}_1\|,
\ee
where constant $C_1=C_1(M_1)>0$ is independent of $j$. Moreover for any $\psi\in L^p(\Gamma)$, there exists $C_2=C_2(M_1,\mu,\nu)$ such that
\be
\left\|\mathcal{J}_j\psi^{\otimes j}\right\|_{L^q(B_a)}\le
C_2\left(\mu+2\nu\right)^j\|\mathcal{J}_1\psi\|_{L^q(B_a)}^j.
\ee
\end{thm}

\begin{proof}
We find that for $j\ge2$,
\ba
&\fl
\|\mathcal{J}_j\|=
\left\|\left(\sum_{m=1}^{j-1}\mathcal{J}_m\sum_{i_1+\cdots+i_m=j}
J_{i_1}\otimes\cdots\otimes J_{i_m}\right)
\mathcal{J}_1\otimes\cdots\otimes\mathcal{J}_1\right\|
\\
&\fl\le
\left\|\sum_{m=1}^{j-1}\mathcal{J}_m\nu^m
\sum_{i_1+\cdots+i_m=j}\left(\mu+\nu\right)^{i_1-1}\cdots
\left(\mu+\nu\right)^{i_m-1}\right\|\|\mathcal{J}_1\|^j
\\
&\fl\le
\sum_{m=1}^{j-1}\|\mathcal{J}_m\|\nu^m
\left(\begin{array}{c}j-1\\ m-1\end{array}\right)\left(\mu+\nu\right)^{j-m}\|\mathcal{J}_1\|^j
\\
&\fl\le
\|\mathcal{J}_1\|^j
\left(\sum_{m=1}^{j-1}\|\mathcal{J}_m\|\right)
\left(\sum_{m=1}^{j-1}
\left(\begin{array}{c}j-1\\ m-1\end{array}\right)\nu^m\left(\mu+\nu\right)^{j-m}\right).
\ea
By using (\ref{binomial}), we have
\ba
\|\mathcal{J}_j\|
&\le
\|\mathcal{J}_1\|^j
\left(\sum_{m=1}^{j-1}\|\mathcal{J}_m\|\right)
\left(\nu\left(\mu+2\nu\right)^{j-1}-\nu^j\right)
\\
&\le
\nu\|\mathcal{J}_1\|^j\left(\mu+2\nu\right)^{j-1}
\sum_{m=1}^{j-1}\|\mathcal{J}_m\|
\\
&\le
\|\mathcal{J}_1\|^j\left(\mu+2\nu\right)^j
\sum_{m=1}^{j-1}\|\mathcal{J}_m\|.
\ea

By noticing the recursive structure of the above inequality, we can write
\be
\|\mathcal{J}_j\|\le
c_j\left[\left(\mu+2\nu\right)\|\mathcal{J}_1\|\right]^j
\|\mathcal{J}_1\|,
\ee
where
\be
c_{j+1}=c_j+
\left[\left(\mu+2\nu\right)\|\mathcal{J}_1\|\right]^jc_j,\quad
c_2=1.
\ee
Hence we obtain
\be
c_j=\prod_{m=2}^{j-1}\left(1+\left[\left(\mu+2\nu\right)\|\mathcal{J}_1\|\right]^m\right),\quad j\ge3.
\ee
We note that
\ba
\ln{c_j}
&\le
\sum_{m=1}^{j-1}\ln\left(1+
\left[\left(\mu+2\nu\right)\|\mathcal{J}_1\|\right]^m\right)
\\
&\le
\sum_{m=1}^{j-1}\left[\left(\mu+2\nu\right)\|\mathcal{J}_1\|\right]^m
\\
&\le
\frac{1}{1-\left(\mu+2\nu\right)\|\mathcal{J}_1\|}
\\
&\le
\frac{1}{1-M_1}.
\ea
Thus $c_j$ ($j\ge2$) are bounded. We put $C_1=\exp(1/(1-M_1))$.

We note that
\be
\left\|\mathcal{J}_j\psi^{\otimes j}\right\|_{L^q(B_a)}\le
\|\mathcal{J}_1\psi\|_{L^q(B_a)}^j\left(\mu+2\nu\right)^j
\sum_{m=1}^{j-1}\|\mathcal{J}_m\|,
\ee
and
\be
\sum_{m=1}^{j-1}\|\mathcal{J}_m\|\le
C_1\|\mathcal{J}_1\|\left(\mu+2\nu\right)
\frac{1-\left(\mu+2\nu\right)^{j-1}}{1-\left(\mu+2\nu\right)}.
\ee
Hence we obtain
\ba
&\fl
\left\|\mathcal{J}_j\psi^{\otimes j}\right\|_{L^q(B_a)}
\le
C_1\left(\mu+2\nu\right)^{j+1}\|\mathcal{J}_1\|
\frac{1-\left(\mu+2\nu\right)^{j-1}}{1-\left(\mu+2\nu\right)}
\|\mathcal{J}_1\psi\|_{L^q(B_a)}^j
\\
&\fl\le
\frac{C_1M_1}{1-\left(\mu+2\nu\right)}
\left(\mu+2\nu\right)^j\|\mathcal{J}_1\psi\|_{L^q(B_a)}^j.
\ea
The proof is complete if we set
\be
C_2=\frac{C_1M_1}{1-\left(\mu+2\nu\right)}.
\ee
\end{proof}

Let us consider the convergence of the inverse Rytov series. If the inverse Rytov series converges, we write
\be
\eta\approx\widetilde{\eta},
\ee
where
\be
\widetilde{\eta}=\sum_{j=1}^{\infty}\mathcal{J}_j\psi^{\otimes j}.
\ee

\begin{thm}
Assume that there exists a constant $M_1<1$ such that $(\mu+2\nu)\|\mathcal{J}_1\|\le M_1$. Suppose that $\|\mathcal{J}_1\psi\|_{L^q(B_a)}<(\mu+2\nu)^{-1}$. Let $M_2=\max(\|\eta\|_{L^q(B_a)},\|\mathcal{J}_1J_1\eta\|_{L^q(B_a)})$. We assume that $M_2<(\mu+2\nu)^{-1}$. Then for any $N\in\Nm$ there exists constants $C_3=C_3(M_1,M_2,\mu,\nu)>0$ such that
\be\fl
\left\|\eta-\sum_{j=1}^N\mathcal{J}_j\psi^{\otimes j}\right\|_{L^q(B_a)}
\le
C_3\|(I-\mathcal{J}_1J_1)\eta\|_{L^q(B_a)}+C_2
\frac{\left[(\mu+2\nu)\|\mathcal{J}_1\psi\|_{L^q(B_a)}\right]^{N+1}}{1-(\mu+2\nu)\|\mathcal{J}_1\psi\|_{L^q(B_a)}},
\ee
where constant $C_2>0$ is given in Theorem \ref{thm4_1}.
\end{thm}

\begin{proof}
If we expand $\psi$ in the inverse Rytov series by the Rytov series, we can write
\be
\widetilde{\eta}=\sum_{j=1}^{\infty}\widetilde{\mathcal{J}}_j\eta\otimes\cdots\otimes\eta,
\ee
where
\be
\widetilde{\mathcal{J}}_1=\mathcal{J}_1J_1,
\ee
and
\be\fl
\widetilde{\mathcal{J}}_j=
\left(\sum_{m=1}^{j-1}\mathcal{J}_m\sum_{i_1+\cdots+i_m=j}J_{i_1}\otimes\cdots\otimes J_{i_m}\right)+\mathcal{J}_jJ_1\otimes\cdots\otimes J_1,
\quad j\ge2.
\ee
We have
\be
\widetilde{\mathcal{J}}_j=
\sum_{m=1}^{j-1}\mathcal{J}_m\sum_{i_1+\cdots+i_m=j}J_{i_1}\otimes\cdots\otimes J_{i_m}
\left(I-\mathcal{J}_1J_1\otimes\cdots\otimes\mathcal{J}_1J_1\right).
\ee

Since
\be
\eta-\widetilde{\eta}=
(I-\mathcal{J}_1J_1)\eta-
\mathcal{J}_1J_2\left(\eta\otimes\eta-\mathcal{J}_1J_1\eta\otimes\mathcal{J}_1J_1\eta\right)+\cdots,
\ee
we have
\ba
\left\|\eta-\widetilde{\eta}\right\|_{L^q(B_a)}
&\le
\sum_{j=1}^{\infty}\sum_{m=1}^{j-1}\sum_{i_1+\cdots+i_m=j}
\|\mathcal{J}_m\|\|\mathcal{J}_{i_1}\|\cdots\|\mathcal{J}_{i_m}\|
\\
&\times
\left\|(\eta\otimes\cdots\otimes\eta)-\left(\mathcal{J}_1J_1\eta\otimes\cdots\otimes\mathcal{J}_1J_1\eta\right)\right\|_{L^q(B_a^j)}.
\ea
We note the identity
\ba
&
\left(\eta_1\otimes\cdots\otimes\eta_1\right)-
\left(\eta_2\otimes\cdots\otimes\eta_2\right)
\\
&=
(\eta_1-\eta_2)\otimes\eta_2\otimes\cdots\otimes\eta_2+
\eta_1\otimes(\eta_1-\eta_2)\otimes\eta_2\otimes\cdots\otimes\eta_2
+\cdots
\\
&+
\eta_1\otimes\eta_1\otimes\cdots\otimes(\eta_1-\eta_2)\otimes\eta_2+
\eta_1\otimes\eta_1\otimes\cdots\otimes\eta_1\otimes(\eta_1-\eta_2).
\ea
Hence,
\be\fl
\left\|\eta\otimes\cdots\otimes\eta-\mathcal{J}_1J_1\eta\otimes\cdots\otimes\mathcal{J}_1J_1\eta\right\|_{L^q(B_a^j)}
\le
jM_2^{j-1}\left\|\eta-\mathcal{J}_1J_1\eta\right\|_{L^q(B_a)}.
\ee
We obtain
\be\fl
\left\|\eta-\widetilde{\eta}\right\|_{L^q(B_a)}\le
\sum_{j=1}^{\infty}\sum_{m=1}^{j-1}\sum_{i_1+\cdots+i_m=j}
\|\mathcal{J}_m\|\|\mathcal{J}_{i_1}\|\cdots\|\mathcal{J}_{i_m}\|
jM_2^{j-1}\left\|\eta-\mathcal{J}_1J_1\eta\right\|_{L^q(B_a)}.
\ee
Furthermore,
\ba
&\fl
\left\|\eta-\widetilde{\eta}\right\|_{L^q(B_a)}
\\
&\fl\le
\sum_{j=1}^{\infty}\sum_{m=1}^{j-1}jM_2^{j-1}\|\mathcal{J}_m\|
\left(\begin{array}{c}j-1\\m-1\end{array}\right)\nu^m
\left(\mu+\nu\right)^{j-m}
\left\|\eta-\mathcal{J}_1J_1\eta\right\|_{L^q(B_a)}
\\
&\fl\le
\left\|\eta-\mathcal{J}_1J_1\eta\right\|_{L^q(B_a)}
\sum_{j=1}^{\infty}jM_2^{j-1}\left(\sum_{m=1}^{j-1}\|\mathcal{J}_m\|\right)
\left(\sum_{m=1}^{j-1}\left(\begin{array}{c}j-1\\m-1\end{array}\right)
\nu^m\left(\mu+\nu\right)^{j-m}\right)
\\
&\fl=
\nu\left\|\eta-\mathcal{J}_1J_1\eta\right\|_{L^q(B_a)}
\sum_{j=1}^{\infty}jM_2^{j-1}\left(\sum_{m=1}^{j-1}\|\mathcal{J}_m\|\right)
\left[\left(\mu+2\nu\right)^{j-1}-\nu^{j-1}\right]
\\
&\fl\le
\left\|\eta-\mathcal{J}_1J_1\eta\right\|_{L^q(B_a)}
\sum_{j=1}^{\infty}jM_2^{j-1}\left(\mu+2\nu\right)^j
\left(\sum_{m=1}^{j-1}\|\mathcal{J}_m\|\right).
\ea
Using $C_1>0$ in Theorem \ref{thm4_1}, we obtain
\ba
&\fl
\left\|\eta-\widetilde{\eta}\right\|_{L^q(B_a)}
\le
C_1\|\mathcal{J}_1\|\left\|\eta-\mathcal{J}_1J_1\eta\right\|_{L^q(B_a)}
\sum_{j=1}^{\infty}jM_2^{j-1}\left(\mu+2\nu\right)^{j+1}
\frac{1-\left(\mu+2\nu\right)^{j-1}}{1-\left(\mu+2\nu\right)}
\\
&\fl\le
C_1\left\|\eta-\mathcal{J}_1J_1\eta\right\|_{L^q(B_a)}
\frac{\mu+2\nu}{1-\left(\mu+2\nu\right)}
\sum_{j=1}^{\infty}j\left[M_2\left(\mu+2\nu\right)^{j-1}\right]
\\
&\fl=
C_3\left\|\eta-\mathcal{J}_1J_1\eta\right\|_{L^q(B_a)},
\ea
where
\be
C_3=C_1\frac{\mu+2\nu}{1-\left(\mu+2\nu\right)}
\sum_{j=1}^{\infty}j\left[M_2\left(\mu+2\nu\right)^{j-1}\right].
\ee

We have
\be
\left\|\widetilde{\eta}\right\|_{L^q(B_a)}\le
\sum_{j=1}^{\infty}\|\mathcal{J}_j\psi^{\otimes j}\|_{L^q(B_a)}
\le
C_2\sum_{j=1}^{\infty}\left(\mu+2\nu\right)^j\|\mathcal{J}_1\psi\|_{L^q(B_a)}^j.
\ee
Hence $\widetilde{\eta}$ converges. We note that
\ba
\left\|\widetilde{\eta}-\sum_{j=1}^N\mathcal{J}_j\psi^{\otimes j}\right\|_{L^q(B_a)}
&\le
\sum_{j=N+1}^{\infty}\left\|\mathcal{J}_j\psi\otimes\cdots\otimes\psi\right\|_{L^q(B_a)}
\\
&\le
C_2\sum_{j=Ns+1}^{\infty}\left(\mu+2\nu\right)^j\|\mathcal{J}_1\psi\|_{L^q(B_a)}^j
\\
&=
C_2\frac{\left[\left(\mu+2\nu\right)\|\mathcal{J}_1\psi\|_{L^q(B_a)}\right]^{N+1}}
{1-\left(\mu+2\nu\right)\|\mathcal{J}_1\psi\|_{L^q(B_a)}}
\ea
The proof is complete.
\end{proof}

The stability of the reconstruction is studied as follows.

\begin{thm}
Assume that there exists a constant $M_1<1$ such that $(\mu+2\nu)\|\mathcal{J}_1\|\le M_1$. Let $\eta_1,\eta_2$ denote the limits of the inverse Rytov series corresponding to some $\psi_1,\psi_2$. We suppose that $M_1M_3<1$, where $M_3=\max(\|\psi_1\|_{L^p(\Gamma)},\|\psi_2\|_{L^p(\Gamma)})$. Then there exists $C_4=C_4(M_1,M_3,\mu,\nu)>0$ such that
\be
\|\eta_1-\eta_2\|_{L^q(B_a)}<C_4\|\psi_1-\psi_2\|_{L^p(\Gamma)}.
\ee
\end{thm}

\begin{proof}
We begin with the following inequality.
\be
\|\eta_1-\eta_2\|_{L^q(B_a)}\le
\sum_{j=1}^{\infty}\left\|\mathcal{J}_j\psi_1\otimes\cdots\otimes\psi_1-
\mathcal{J}_j\psi_2\otimes\cdots\otimes\psi_2\right\|_{L^q(B_a)}.
\ee
We note that
\ba
&
(\psi_1\otimes\cdots\otimes\psi_1)-(\psi_2\otimes\cdots\otimes\psi_2)
\\
&=
(\psi_1-\psi_2)\otimes\psi_2\otimes\cdots\otimes\psi_2+
\psi_1\otimes(\psi_1-\psi_2)\otimes\psi_2\otimes\cdots\otimes\psi_2+\cdots
\\
&+
\psi_1\otimes\cdots\otimes\psi_1\otimes(\psi_1-\psi_2)\otimes\psi_2+
\psi_1\otimes\cdots\otimes\psi_1\otimes(\psi_1-\psi_2).
\ea
We obtain
\ba
&\fl
\|\eta_1-\eta_2\|_{L^q(B_a)}
\\
&\fl\le
\sum_{j=1}^{\infty}\|\mathcal{J}_j\|\sum_{k=1}^j
\|\psi_1\|_{L^p(\Gamma)}\cdots\|\psi_1\|_{L^p(\Gamma)}\|(\psi_1-\psi_2)\|_{L^p(\Gamma)}\|\psi_2\|_{L^p(\Gamma)}\cdots\|\psi_2\|_{L^p(\Gamma)},
\ea
where $\|(\psi_1-\psi_2)\|_{L^p(\Gamma)}$ is in the $k$th position of the product. Furthermore,
\ba
\|\eta_1-\eta_2\|_{L^q(B_a)}
&\le
\sum_{j=1}^{\infty}j\|\mathcal{J}_j\|M_3^{j-1}\|\psi_1-\psi_2\|_{L^p(\Gamma)}
\\
&\le
C_1\|\mathcal{J}_1\|\|\psi_1-\psi_2\|_{L^p(\Gamma)}\sum_{j=1}^{\infty}j\left(\mu+2\nu\right)^jM_3^{j-1}
\\
&\le
\frac{C_1}{M_3}\|\psi_1-\psi_2\|_{L^p(\Gamma)}\sum_{j=1}^{\infty}j\left(\mu+2\nu\right)^{j-1}M_3^{j-1}.
\ea
The proof is complete if we put
\be
C_4=C_1\sum_{j=1}^{\infty}j\left(\mu+2\nu\right)^{j-1}M_3^{j-2}.
\ee
\end{proof}

\section{Two-dimensional radial problem}
\label{toymodel}

\subsection{Setup}

Let us assume the two-dimensional radial geometry, which was considered in \cite{Moskow-Schotland09}. We consider diffuse optical tomography for this domain. In the polar coordinate system we have $x=(r,\theta)$, where $r$ is the radial coordinate and $\theta$ is the angular coordinate. Let $\Omega$ be the disk of radius $R$ centered at the origin. Assuming that $\eta$ has the radial symmetry, we write
\be
\eta(x)=\eta(r),\quad 0<r<R.
\ee
Let us suppose that $\eta$ is given by
\be
\eta(r)=\cases{
\eta_a,&$0\le r\le R_a$,
\\
0,&$R_a<r\le R$.}
\ee

Although point sources were used in \cite{Moskow-Schotland09}, here we assume the following spatially oscillating source term for diffuse optical tomography in spatial frequency domain.
\be
f(r,\theta)=e^{i\alpha\theta}\frac{1}{r}\delta(r-R),\quad\alpha=1,\dots,M_S.
\ee
Hereafter we write
\be
g=k^2,\quad k>0.
\ee
We define $\ell=\zeta D_0$ and set $D_0=1$. We write $\Omega_1=\{x;\;|x|\le R_a\}$, $\Omega_2=\{x;\;R_a<|x|<R\}$, and $k_a=k\sqrt{1+\eta_a}$. Let $r_x,r_y$ be the radial coordinates of $x,y$. Let $\theta_x,\theta_y$ be the angular coordinates of $x,y$.

\subsection{Forward problem}

Let us express the Green's function $G(x,y)$, which has the source term $\frac{1}{r_x}\delta(r_x-r_y)\delta(\theta_x-\theta_y)$, as \cite{Markel-Schotland02}
\be
G(x,y)=\frac{1}{2\pi}\sum_{n=-\infty}^{\infty}e^{in(\theta_x-\theta_y)}g_n(r_x,r_y),
\ee
where $g_n(r,r')$ satisfies
\ba
r^2\pp_r^2g_n(r,r')+r\pp_rg_n(r,r')-\left(k^2r^2+n^2\right)g_n(r,r')=-r\delta(r-r'),
\\
g_n(R,r')+\ell\pp_rg_n(R,r')=0.
\ea
We note that the homogeneous equation for the above equation is the modified Bessel differential equation. Hence the solution $u$ is given as a superposition of $I_n(kr),K_n(kr)$. Here, $I_n,K_n$ are the modified Bessel functions of the first and second kinds, respectively. We obtain
\ba
g_n(r_x,r_y)
&=
K_n\left(k\max(r_x,r_y)\right)I_n\left(k\min(r_x,r_y)\right)
\\
&-
\frac{K_n(kR)+k\ell K'_n(kR)}{I_n(kR)+k\ell I'_n(kR)}I_n\left(kr_x\right)I_n\left(kr_y\right).
\ea
We note
\be\fl
I_n'(x)=\frac{1}{2}\left(I_{n-1}(x)+I_{n+1}(x)\right),\quad
K_n'(x)=-\frac{1}{2}\left(K_{n-1}(x)+K_{n+1}(x)\right),\quad x\in\Rm.
\ee
Hence,
\be
u_0(x)=\int_0^{2\pi}\int_0^RG(x,y)f(r_y,\theta_y)r_y\,dr_yd\theta_y=
e^{i\alpha\theta_x}g_{\alpha}(r_x,R).
\ee
We have
\be
g_{\alpha}(R,R)=I_{\alpha}(kR)K_{\alpha}(kR)-d_{\alpha}I_{\alpha}(kR),
\ee
where
\be
d_{\alpha}=\frac{K_{\alpha}(kR)+k\ell K'_{\alpha}(kR)}{I_{\alpha}(kR)+k\ell I'_{\alpha}(kR)}I_{\alpha}(kR).
\ee

For the forward data, we observe $u,u_0$ at $r_x=R$, $\theta_x=0$. That is, the outgoing light is measured at one point on the boundary while boundary values were observed at different points on the boundary in \cite{Moskow-Schotland09}. See \ref{fwd} for the calculation of $u$. Let us set $M_D=1$ (i.e., $M_{\rm SD}=M_S$). For the vector $\vv{\psi}\in\Rm^{M_{\rm SD}}$, we have
\begin{equation}
\psi_{\alpha}=\ln\frac{\left(K_{\alpha}(kR)-d_{\alpha}\right)I_{\alpha}(kR)}{I_{\alpha}(kR)K_{\alpha}(kR)+b_{\alpha}K_{\alpha}(kR)+c_{\alpha}I_{\alpha}(kR)}
\label{fdata}
\end{equation}
for $\alpha=1,\dots,M_{\rm SD}$. We note that $b_{\alpha},c_{\alpha}$, which are given in \ref{fwd}, depend on $\eta_a,R_a,k,R,\ell$.

Let us introduce
\be
G^{(n)}(r_x,r_y):=g_n(r_x,r_y)r_y.
\ee
We obtain
\ba
&\fl\left(K_j\eta^{\otimes j}\right)(x)
\\
&\fl=
(-1)^{j+1}g^je^{i\alpha\theta_x}\int_0^R\cdots\int_0^R
G^{(\alpha)}(R,r_{y_1})G^{(\alpha)}(r_{y_1},r_{y_2})\cdots
G^{(\alpha)}(r_{y_{j-1}},r_{y_j})G^{(\alpha)}(r_{y_j},R)
\\
&\fl\times 
\eta(r_{y_1})\cdots\eta(r_{y_j})\,dr_{y_1}\cdots dr_{y_j},\quad x\in\Gamma,
\ea
and
\ba
&\fl\left(\check{K}_j\eta^{\otimes j}\right)(x)=
\left(\frac{1}{u_0}K_j\eta^{\otimes j}\right)(x)
\\
&\fl=
\frac{(-1)^{j+1}g^j}{G^{(\alpha)}(R,R)}\int_0^R\cdots\int_0^R
G^{(\alpha)}(R,r_{y_1})G^{(\alpha)}(r_{y_1},r_{y_2})\cdots
G^{(\alpha)}(r_{y_{j-1}},r_{y_j})G^{(\alpha)}(r_{y_j},R)
\\
&\fl\times 
\eta(r_{y_1})\cdots\eta(r_{y_j})\,dr_{y_1}\cdots dr_{y_j},
\quad x\in\Gamma.
\ea

\subsection{Implementation of the inverse Rytov series}
\label{implem}

Let us begin by writing
\be
\psi_{\alpha}=
\sum_{j=1}^{\infty}\left(J_j^{(\alpha)}\eta^{\otimes j}\right)(R),
\quad\alpha=1,\dots,M_{\rm SD}.
\ee
We consider how the $j$th-order operator $\mathcal{J}_j$ in the inverse Rytov series can be numerically constructed. Here we assume that $r\in(0,R)$ is discretized into $N_r$ points $r_i$ ($i=1,\dots,N_r$) with small interval $\Delta r$. Thus, $\eta$ can be expressed by a vector $\vv{\eta}\in\Rm^{N_r}$.

\subsubsection{Forward vectors}

We set
\be
r_i=i\Delta r\quad(i=1,\dots,N_r),\quad\Delta r=\frac{R}{N_r}.
\ee
Let $\vv{b}\in\Rm^{N_r}$ be a vector. We define $\vv{K}_0\in\Rm^{M_{\rm SD}}$, $\vv{K}_1\in\Rm^{M_{\rm SD}N_r}$ as
\ba
&
\{\vv{K}_0\}_{\alpha}=-G^{(\alpha)}(R,R),
\\
&
\{\vv{K}_1(\vv{b})\}_{i+(\alpha-1)N_r}=g\Delta r
\sum_{n=1}^{N_r}G^{(\alpha)}(r_i,r_n)G^{(\alpha)}(r_n,R)\{\vv{b}\}_n
\ea
for $1\le\alpha\le M_{\rm SD}$, $1\le i\le N_r$. Moreover,
\ba
&\{\vv{K}_j(\vv{b}_1,\dots,\vv{b}_j)\}_{i+(\alpha-1)N_r}
\\
&=-g\Delta r\sum_{n=1}^{N_r}G^{(\alpha)}(r_i,r_n)\{\vv{b}_j\}_n
\{\vv{K}_{j-1}(\vv{b}_1,\dots,\vv{b}_{j-1})\}_{n+(\alpha-1)N_r}.
\ea
Using there $\vv{K}_j\in\Rm^{M_{\rm SD}N_r}$, we introduce
\ba
&\fl\left\{\vv{J}_j(\vv{b}_1,\dots,\vv{b}_j)\right\}_{\alpha}
=\sum_{m=1}^j\frac{(-1)^m}{m\{\vv{K}_0\}_{\alpha}^m}
\\
&\fl\times
\sum_{i_1+\cdots i_m=j}
\left\{\vv{K}_{i_1}(\vv{b}_1,\dots,\vv{b}_{i_1})\right\}_{\alpha N_r}
\cdots
\left\{\vv{K}_{i_m}(\vv{b}_{j-i_m+1},\dots,\vv{b}_j)\right\}_{\alpha N_r}
\ea
for $\alpha=1,\dots,M_{\rm SD}$.

\subsubsection{Linearized problem}

In particular, we have
\be
\left\{\vv{J}_1(\vv{b})\right\}_{\alpha}
=-\frac{1}{\{\vv{K}_0\}_{\alpha}}\left\{\vv{K}_1(\vv{b})\right\}_{\alpha N_r}
\ee
for $\alpha=1,\dots,M_{\rm SD}$, $i=1,\dots,N_r$. From this, we can define a matrix $\underline{J}_1\in\Rm^{M_{\rm SD}\times N_r}$ such that $\vv{J}_1(\vv{b})=\underline{J}_1\vv{b}$ as
\be
\left\{\underline{J}_1\right\}_{\alpha,i}=
\frac{g\Delta r}{G^{(\alpha)}(R,R)}\left[G^{(\alpha)}(R,r_i)\right]^2.
\ee

Using $\underline{J}_1$, we compute $\underline{\mathcal{J}}_1$. Here, $\underline{\mathcal{J}}_1$ is the Moore-Penrose pseudoinverse with a regularizer such as the truncated singular value decomposition:
\be
\underline{\mathcal{J}}_1=\underline{J}_{1,{\rm reg}}^+\in\Rm^{N_r\times M_{\rm SD}}.
\ee
The first term of the inverse Rytov series can be calculated as
\be
\vv{\eta}_1=\underline{\mathcal{J}}_1\vv{\psi},
\ee
where
\be
\{\vv{\eta}_1\}_i=\eta_1(r_i),\quad i=1,\dots,N_r.
\ee
We solve $\vv{\psi}=\underline{J}_1\vv{\eta}_1$ as follows.

\paragraph{Underdetermined}

Suppose we have
\be
M_{\rm SD}\le N_r.
\ee
That is, the inverse problem is underdetermined.

In this case, we obtain
\be
\vv{\eta}_1=\underline{J}_{1,{\rm reg}}^+\vv{\psi},
\ee
where
\be
\underline{J}_{1,{\rm reg}}^+=\underline{J}_1^*\underline{M}_{\rm reg}^{-1},
\quad \underline{M}=\underline{J}_1\underline{J}_1^*.
\ee
Here, $*$ denotes the Hermitian conjugate and ${\rm reg}$ means that the pesudoinverse is regularized by discarding singular values that are smaller than $\sigma_0$. Let $\sigma_j^2$ and $\vv{v}_j$ be the eigenvalues and eigenvectors of the matrix $\underline{M}$:
\be
\underline{M}\vv{z}_j=\sigma_j^2\vv{z}_j.
\ee
We obtain
\be
\vv{\eta}_1=\sum_{j\atop\sigma_j>\sigma_0}\frac{1}{\sigma_j^2}
\left(\vv{z}_j^*\vv{\psi}\right)\underline{J}_1^*\vv{z}_j.
\ee

\paragraph{Overdetermined}

Suppose we have
\be
M_{\rm SD}\ge N_r.
\ee
That is, the inverse problem is overdetermined.

In this case, we obtain
\be
\vv{\eta}_1=\underline{J}_{1,{\rm reg}}^+\vv{\psi},
\ee
where
\be
\underline{J}_{1,{\rm reg}}^+=\underline{M}_{\rm reg}^{-1}\underline{J}_1^*,
\quad \underline{M}=\underline{J}_1^*\underline{J}_1.
\ee
After solving the eigenproblem $\underline{M}\vv{z}_j=\sigma_j^2\vv{z}_j$, we obtain
\be
\vv{\eta}_1=\sum_{j\atop\sigma_j>\sigma_0}\frac{1}{\sigma_j^2}
\left(\vv{z}_j^*\underline{J}_1^*\vv{\psi}\right)\vv{z}_j.
\ee

\subsubsection{Inversion}

Let $\vv{a}_1,\dots,\vv{a}_j$ be real vectors of dimension $M_{\rm SD}$. To compute the $j$th-order term $\vv{\eta}_j$, let us first introduce
\be
\vv{\eta}^{(1)}_i=\underline{\mathcal{J}}_1\vv{a}_i\quad (i=1,\dots,j).
\ee
We introduce vector $\vv{\mathcal{J}}_j(\vv{a}_1,\dots,\vv{a}_j)\in\Rm^{N_r}$ which has a recursive structure: for $j=1$,
\be
\vv{\mathcal{J}}_1(\vv{a}_1)=\vv{\eta}^{(1)}_1,
\ee
and for $j\ge2$,
\ba
&
\vv{\mathcal{J}}_j(\vv{a}_1,\dots,\vv{a}_j)=
-\sum_{m=1}^{j-1}\sum_{i_1+\cdots+i_m=j}
\\
&
\vv{\mathcal{J}}_m\left(
\vv{J}_{i_1}(\vv{\eta}^{(1)}_1,\dots,\vv{\eta}^{(1)}_{i_1}),\cdots,
\vv{J}_{i_m}(\vv{\eta}^{(1)}_{j-i_m+1},\dots,\vv{\eta}^{(1)}_j)\right).
\ea
More specifically, $\vv{\mathcal{J}}_j(\vv{a}_1,\dots,\vv{a}_j)$ can be computed as follows. If $j=1$, then $\vv{\eta}^{(1)}_1$ is returned. For $j\ge2$, we let $m$ move from $1$ to $j-1$. We form the compositions $[i_1,\dots,i_m]$ such that $i_1+\cdots+i_m=j$. For each $m$ ($1\le m\le j-1$) and each composition $(i_1,\dots,i_m)$, we compute
\be
\vv{\eta}_{\rm tmp}=
-\vv{\mathcal{J}}_m\left(
\vv{J}_{i_1}(\vv{\eta}^{(1)}_1,\dots,\vv{\eta}^{(1)}_{i_1}),\cdots,
\vv{J}_{i_m}(\vv{\eta}^{(1)}_{j-i_m+1},\dots,\vv{\eta}^{(1)}_j)\right).
\ee
Let $\vv{\Sigma}(m)$ denote the sum of $\vv{\eta}_{\rm tmp}$ for all $\left(\begin{array}{c}j-1\\ m-1\end{array}\right)$ compositions. The above step is repeated for all $m$ ($1\le m\le j-1$). We obtain
\be
\vv{\mathcal{J}}_j(\vv{a}_1,\dots,\vv{a}_j)=\sum_{m=1}^{j-1}\vv{\Sigma}(m).
\ee

The $j$th term is calculated as
\be
\vv{\eta}_j=\vv{\mathcal{J}}_j(\vv{\psi},\dots,\vv{\psi}).
\ee
In this way, we obtain $\vv{\eta}_j$ ($j=1,\dots,N$). The $N$th-order approximation is given by
\be
\vv{\eta}^{(N)}=\vv{\eta}_1+\cdots+\vv{\eta}_N.
\ee

Finally, the reconstruction can be done as follows.We have $\mu_a(r)=g(1+\eta(r))$. The reconstructed $\mu_a(r)$ is obtained as
\be
\mu_a(r_i)\approx \mu_a^{(N)}(r_i)
=g\left(1+\left\{\vv{\eta}^{(N)}\right\}_i\right).
\ee

\subsection{Numerical results}

We set $k=1$, $R=3$, $R_a=1.5$, $\ell=0.3$. Moreover, $N_r=M_{\rm SD}=90$. We chose $\sigma_0$ such that the largest $23$ singular values were taken. Since only $23$ singular values are taken, $\eta$ is not fully reconstructed. When reconstructing $\eta$, we obtain at most $\vv{\eta}_{\rm proj}\in\Rm^{N_r}$, which is given by
\be
\vv{\eta}_{\rm proj}=\underline{\mathcal{J}}_1\underline{J}_1\vv{\eta}.
\ee
In Figs.~\ref{fig1} through \ref{fig3}, $\vv{\eta},\vv{\eta}_{\rm proj},\vv{\eta}^{(1)},\vv{\eta}^{(2)},\vv{\eta}^{(3)},\vv{\eta}^{(4)},\vv{\eta}^{(5)}$ are shown.

In Fig.~\ref{fig1}, $\eta_a=0.2$ (Left) and $1.0$ (Right), and in Fig.~\ref{fig2}, $\eta_a=2.0$ (Left) and $5.0$ (Right). In the case of $\eta_a=0.2$, the first Rytov approximation $\vv{\eta}^{(1)}$ is different from $\vv{\eta}_{\rm proj}$ but the third Rytov approximation $\vv{\eta}^{(3)}$ already gives a good reconstruction. For $\eta_a=1.0$, the reconstruction approaches $\eta_{\rm proj}$ after the fifth term $\vv{\eta}_5$ is added. When $\eta_a=2.0$, the reconstruction is reasonable after $\vv{\eta}_5$ is added but still different from $\vv{\eta}_{\rm proj}$. When $\eta_a$ is large and $\eta_a=5.0$, all reconstructions differ from $\vv{\eta}_{\rm proj}$.

\begin{figure}[t]
\begin{center}
\includegraphics[width=0.45\textwidth]{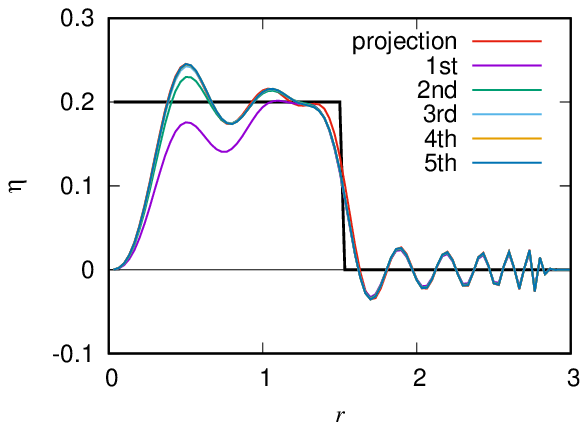}
\hspace{5mm}
\includegraphics[width=0.45\textwidth]{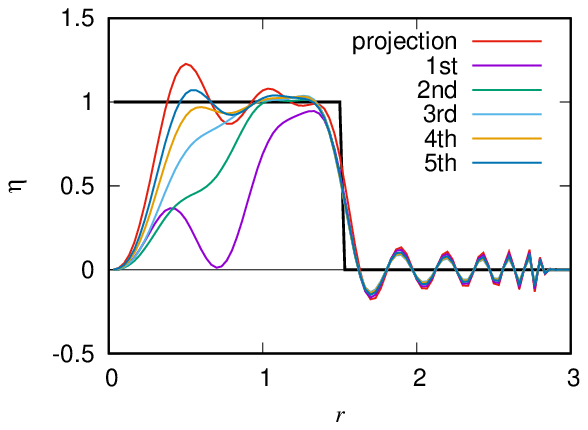}
\vspace{10mm}
\end{center}
\caption{Reconstruction of $\eta$. The forward data is given in (\ref{fdata}). We set (Left) $\eta_a=0.2$ and (Right) $\eta_a=1$.}
\label{fig1}
\end{figure}

\begin{figure}[t]
\begin{center}
\includegraphics[width=0.45\textwidth]{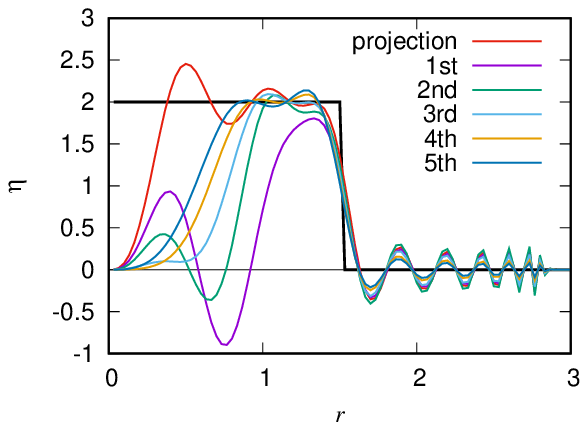}
\hspace{5mm}
\includegraphics[width=0.45\textwidth]{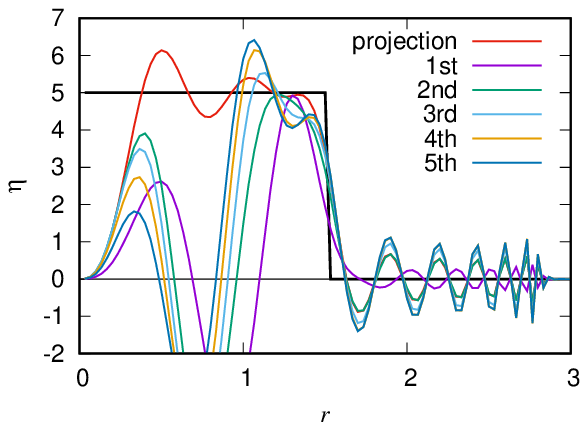}
\vspace{10mm}
\end{center}
\caption{Reconstruction of $\eta$. The forward data is given in (\ref{fdata}). We set (Left) $\eta_a=2$, and (Right) $\eta_a=5$.}
\label{fig2}
\end{figure}

Noise was added for Fig.~\ref{fig3}. For both $u_0,u$, Gaussian noise with mean zero was added. The standard deviation of the noise was the standard deviation of $u_0$ multiplied by a constant $\gamma$. For $u_0,u$, no noise was added when the resulting value became negative. Figure \ref{fig3} shows the reconstruction of $\eta$ for $\eta_a=1.0$. Due to noise, fewer numbers of singular values had to be used. The largest $9$ singular values were used for $\gamma=10^{-4}$ and the largest $7$ singular values were used when $\gamma=10^{-5}$.

\begin{figure}[t]
\begin{center}
\includegraphics[width=0.45\textwidth]{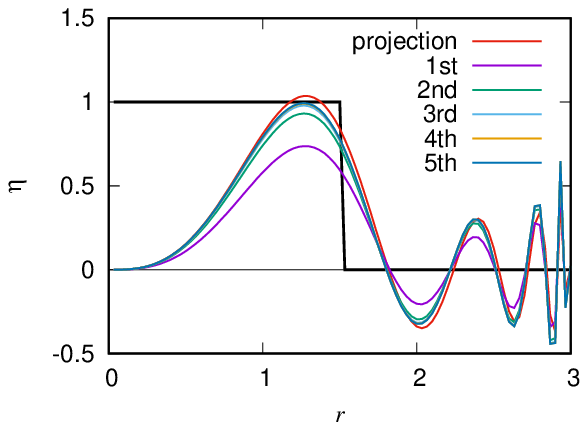}
\hspace{5mm}
\includegraphics[width=0.45\textwidth]{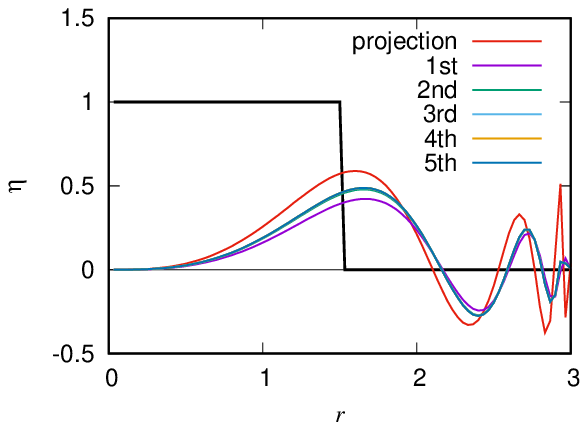}
\vspace{10mm}
\end{center}
\caption{Reconstruction of $\eta$ when $\eta_a=1.0$. Gaussian noise with (Left) $\gamma=10^{-4}$ and (Right) $\gamma=10^{-5}$ was added. The largest $9$ and $7$ singular values were used for the weaker and stronger noise levels, respectively.}
\label{fig3}
\end{figure}

\section{Concluding remarks}
\label{concl}

In this paper, multilinear forward operators $J_j:L^q(B_a)\times\cdots\times L^q(B_a)\to L^p(\Gamma)$ and inverse operators $\mathcal{J}_j:L^p(\Gamma)\times\cdots\times L^p(\Gamma)\to L^q(B_a)$ were considered. As was done for the inverse Born series \cite{Bardsley-Vasquez14,Hoskins-Schotland22,Lakhal18,Moskow-Schotland19}, it is possible to consider the inverse Rytov for nonlinear inverse problems in Banach spaces $X,Y$, for which the forward problem is from $X$ to $Y$ instead of from $L^q(B_a)$ to $L^p(\Gamma)$.

Although the expression of $\psi_j$ in the Rytov series is more complicated than that of $u_j$ in the Born series, the inverse Rytov series can be computed in a recursive manner.

In this paper, the diffusion coefficient $D_0$ was assumed to be a known constant. Markel and Schotland has discussed the simultaneous reconstruction of the two functions with the (first) Rytov approximation \cite{Markel-Schotland04}. It is an interesting future issue to extend the inverse Rytov series to the case of simultaneous reconstruction.

\section*{Acknowledgments}

This work was supported by JST, PRESTO Grant Number JPMJPR2027.

\appendix


\section{Forward data}
\label{fwd}

Let $G_a$ be the Green's function of the two-dimensional radial problem for the equation in which $\eta=\eta_a$ in $r\in[0,R_a]$ and $\eta=0$ otherwise. In the case of the delta-function source $\delta(x-x_s)$, $x_s\in\pp\Omega$, we have \cite{Moskow-Schotland09}
\be
G_a(x,x_s)=\frac{1}{2\pi}\sum_{n=-\infty}^{\infty}a_ne^{in(\theta-\theta_s)}I_n(k_ar),\quad
r\in\Omega_1,
\ee
and
\ba
G_a(x,x_s)
&=
\frac{1}{2\pi}\sum_{n=-\infty}^{\infty}e^{in(\theta-\theta_s)}I_n(kr)K_n(kR)
\\
&+
\frac{1}{2\pi}\sum_{n=-\infty}^{\infty}e^{in(\theta-\theta_s)}\left(b_nK_n(kr)+c_nI_n(kr)\right),\quad r\in\Omega_2.
\ea
Here, coefficients $a_n,b_n,c_n$ can be computed as the solution to the following system of linear equations, which is derived from the interface and boundary conditions.
\ba
&
\left(\begin{array}{ccc}
I_n(k_aR_a) & -K_n(kR_a) & -I_n(kR_a) \\
k_aI_n'(k_aR_a) & -kK_n'(kR_a) & -kI_n'(kR_a) \\
0 & K_n(kR)+k\ell K_n'(kR) & I_n(kR)+k\ell I_n'(kR)
\end{array}\right)
\left(\begin{array}{c}
a_n \\ b_n \\ c_n
\end{array}\right)
\\
&=
\left(\begin{array}{c}
I_n(kR_a)K_n(kR) \\
kI_n'(kR_a)K_n(kR) \\
k\ell I_n(kR)K_n'(kR)+I_n(kR)K_n(kR)
\end{array}\right).
\ea
Therefore we obtain for $x\in\pp\Omega$,
\ba
u(x)
&=\int_0^{2\pi}\int_0^RG_a(x,y)f(r_y,\theta_y)r_y\,dr_yd\theta_y
\\
&=
Re^{i\alpha\theta_x}\left[I_{\alpha}(kR)K_{\alpha}(kR)+b_{\alpha}K_{\alpha}(kR)+c_{\alpha}I_{\alpha}(kR)\right].
\ea

\end{document}